\documentclass[preprint]{imsart}
\pdfoutput=1
\RequirePackage[OT1]{fontenc}
\usepackage{amsthm,amsmath,natbib,amssymb,bm,enumerate,float}
\RequirePackage[colorlinks,citecolor=blue,urlcolor=blue]{hyperref}
\usepackage[pdftex]{graphicx}
\usepackage{booktabs}
\usepackage{pgfplots}
\usepackage{threeparttable}

\startlocaldefs
\numberwithin{equation}{section}
\theoremstyle{plain}
\newtheorem{thm}{Theorem}[section]
\newtheorem{lemma}{Lemma}[section]

\theoremstyle{remark}
\newtheorem{remark}{Remark}[section]

\endlocaldefs

\newcommand{\bmX}{\bm{X}}

\newcommand{\bmv}{\bm{v}}
\newcommand{\bmx}{\bm{x}}
\newcommand{\bmy}{\bm{y}}

\newcommand{\bone}{\bm{1}}
\newcommand{\bzero}{\bm{0}}
\newcommand{\bme}{\bm{e}}
\newcommand{\bmep}{\bm{\epsilon}}
\newcommand{\bmbe}{\bm{\beta}}

\allowdisplaybreaks

\begin{document}

\begin{frontmatter}
\title{Noise Addition for Individual Records to Preserve Privacy and Statistical Characteristics: Case Study of Real Estate Transaction Data}
\runtitle{Method for Noise Addition}

\begin{aug}
\author{\fnms{Yuzo} \snm{Maruyama}
\ead[label=e1]{maruyama@csis.u-tokyo.ac.jp}}
\author{\fnms{Ryoko} \snm{Tone}
\ead[label=e2]{ryo-t@ua.t.u-tokyo.ac.jp}}
\and
\author{\fnms{Yasushi} \snm{Asami}
\ead[label=e3]{asami@csis.u-tokyo.ac.jp}}

\runauthor{Y. Maruyama et al.}

\affiliation{The University of Tokyo}
\address{The University of Tokyo \\
\printead{e1,e2,e3}}
\end{aug}

\begin{abstract}
We propose a new method of perturbing a major variable by adding noise
such that results of regression analysis are unaffected. The extent of the perturbation
can be controlled using a single parameter, which eases an actual perturbation
application. On the basis of results of a numerical experiment, we recommend an
appropriate value of the parameter that can achieve both sufficient perturbation to
mask original values and sufficient coherence between perturbed and original data.
\end{abstract}
\end{frontmatter}

\section{Introduction}
\label{sec:intor}
Increasing amounts of information are now circulated because of recent advancements
in digitalization, thereby increasing the importance of protecting personal information.
Information that can identify a person should not be publicized or utilized without the
person's consent. In the case of information regarding real estate, the location can be
identified by combining several information sources, which in turn might be used to
identify a person, such as an owner or resident. Spatial information of the sort that
relates to real estate can be considered to require special protection.

Two factors are important in dealing with privacy-sensitive information.
First, if information is leaked, then the organization responsible for the information risks
receiving compensation claims because of privacy protection failure.
Second, to avoid possible troubles due to potential information leaks, publicized data tend to become very rough or vague to avoid
potential trouble, often hindering the usefulness of real estate analyses aimed at
understanding the market.

A promising way of dealing with this situation is to protect personal information by
adding noise to acquired data. A typical example of sensitive information is transaction
data, which can include transacted prices, real estate or transacting person
characteristics, and information regarding transaction conditions. Publicized data tend
to omit information about characteristics of transacting persons, and hence, such
contents are assumed not to be included in the database. In this case, one of the most
sensitive types of data will be the transacted price. Private information will be
protected if noise is added to the price data. However, tactless noise addition seriously
distorts data analysis results. Therefore, providing a method of adding noise without
distorting data analyses but still protecting privacy is very important. This study is
devoted to proposing and applying such a method, assuming that the main concern of
the analyses is hedonic analysis, i.e., regression analysis with the transacted price
being the response variable.

\cite{takemura2003current} reviewed statistical issues in publicizing individual data. He listed
several methods of protecting personal information, such as (1) direct hiding by making
the information secret, (2) global categorization by organizing values into several
coarse classes, and (3) disturbance by replacing actual values with different ones (such
as swapping by exchanging individual values, the post-randomization method [PRAM],
or the addition of noise). Direct hiding and global categorization are not appropriate for
releasing data for detailed analyses because the resolution of the information can
become very coarse. Disturbance methods are superior in this aspect, although they
usually introduce errors into analyses, and such effects must be carefully examined.

One well-known method of protecting personal information is the statistical disclosure
limitation (SDL) method. SDL is a general term for methods of protecting identification
of personal data by adding perturbations, modifications, or summarization (\cite{shlomo2010releasing}).
The main concern is to reduce identification risk as well as to retain data usability.

Typically, three kinds of methods are often used to reduce identification risk, including
(1) methods of establishing coarse categorization, (2) methods of generating new data
with statistical characteristics similar to those of the original data, and (3) methods of
adding noise to the original data (\cite{karr2006framework}; \cite{oganian2011masking}).

Substantial research has been conducted on methods of establishing coarse
categorization. In particular, population uniqueness, the feature that a combination of
attributes becomes unique in the parent population, has been studied extensively. For
example, \cite{manrique2012estimating} estimated the risk of population
uniqueness for discrete data.

Regarding methods of generating new data, the swapping method, in which categorical
data are probabilistically exchanged, is well known. One such swapping method is
PRAM, which perturbs the exchanging of categorical data (\cite{Gouweleeuw1998}; \cite{willenborg2001elements}). In this method, a transition probability matrix is constructed and then used as the basis for exchanging categorical data, while maintaining the original proportions of the categories.

A variety of methods of adding noise, while carefully maintaining qualitative features,
have been proposed. For example, \cite{oganian2011masking} focused on features such as
the positivities of values and the magnitude relations between pairs of values. They
proposed a method of adding noise such that the positivities of values, mean values,
and variance-covariance matrices remain the same. One remarkable idea is to use
multiplicative noise addition to avoid obtaining negative values. Moreover, they
demonstrated the stability of results after regression analyses. A similar method of
maintaining the characteristics of attributes was proposed by \cite{abowd2001disclosure}.
Another method of adding noise to avoid the risk of identification is to introduce
random noise distributed following a peculiar symmetric distribution with a hole in the
center. With this method, the perturbed value is never close to the original value, and
therefore, the risk of identification is drastically reduced. In the actual application of
this method, the noise distribution is not publicized, hindering analyses using the
distribution (\cite{reiter2012statistical}).

In general, noise addition can influence the quality of subsequent analyses.
\cite{fuller1993masking} noted that noise addition has an influence similar to that of introducing
measurement errors to explanatory variables. Several methods have been devised to
minimize the influence of noise in particular analyses. For example, some methods
maintain the original mean values and variance-covariance matrices (\cite{ting2008random}; \cite{shlomo2008protection}). In our paper, which focuses on regression analysis, a method is proposed in which adding noise produces robust results.

The paper is organized as follows. In Section \ref{sec:theory}, we propose a method of adding noise to
a response variable and show that some important statistics do not change with noise
addition. In Section \ref{sec:numerical}, numerical experiments are conducted to examine how the
results of multivariate analyses, apart from the assumed regression analysis, can
change. Finally, Section \ref{sec:conclusion} concludes with a summary and suggests possible extensions
of our method.

\section{Theoretical results}
\label{sec:theory}
We assume that the $n\times (p+1)$ design matrix $\bmX$ is given by
$(\bm{1}_n,\bmx_1,\dots,\bmx_p)$, where $\bm{1}_n$ is an $n$-dimensional vector of ones, and
the $n$-dimensional response vector is $\bmy$. We also assume that $n$ is sufficiently larger than
$p$ and that the rank of $\bmX$ is $p+1$.
Then the ordinary least squares (OLS) estimator is 
\begin{equation*}
 \hat{\bm{\beta}}=(\hat{\beta}_0,  \hat{\beta}_1,  \dots , \hat{\beta}_p)'
=\left(\bm{X}'\bm{X}\right)^{-1}\bm{X}'\bm{y}.
\end{equation*}
A decomposition of $\bmy$ based on the OLS estimator $\hat{\bm{\beta}}$ is
$\bmy=\hat{\bmy}+\bm{e}$ where 
\begin{align*}
 \hat{\bmy}=\bmX\hat{\bm{\beta}}=\bmX\left(\bm{X}'\bm{X}\right)^{-1}\bm{X}'\bmy
\end{align*}
is the predictive vector 
and 
\begin{align*}
 \bm{e}=\bmy-\hat{\bmy}=(\bm{I}_n-\bmX\left(\bm{X}'\bm{X}\right)^{-1}\bm{X}')\bmy
\end{align*}
is the residual vector.
Then the coefficient of determination defined by
\begin{equation}\label{Rsq}
 R^2=1-\frac{\|\bm{e}\|^2}{\|\bm{y}-\bar{y}\bm{1}_n\|^2}
\end{equation}
where $\bar{y}$ is the sample mean of $\bmy$,
measures the goodness of fit resulting from the use of the OLS estimator $\hat{\bm{\beta}}$.
The coefficient of determination, $R^2$, is hence regarded as a key quantity in regression analysis.
The $t$-value of the regression coefficient $\beta_j$ for $j=0,1,\dots,p$,
is another key quantity and is defined by
\begin{equation}\label{t-value}
t_j=\frac{\sqrt{n-p-1}}{d_{j}} \frac{\hat{\beta}_j}{\|\bm{e}\|}
\end{equation}
where $d_{j}$ is the $(j+1)$-th diagonal component of $\left(\bm{X}'\bm{X}\right)^{-1}$.
When Gaussian linear regression is performed, $t_j$ has a Student's $t$-distribution
with $n-p-1$ degrees of freedom under the null hypothesis $\beta_j=0$.

The objective of the derivation presented herein is to add perturbation to the original
response vector and achieve tractable tuning of the $R^2$ and $t$-values. Any $n$-dimensional
random vector
\begin{align*}
 \bm{v}=(v_1,\dots,v_n)'.
\end{align*}
may be used as the starting point. Since $n$ is sufficiently
greater than $p$, $\bm{v}$ cannot be expressed as a linear combination of
$\bm{e}, \bm{1}_n,\bm{x}_1,\dots,\bm{x}_p$ with probability one.
In other words,
\begin{equation}\label{def:u}
\bm{u}= \left(\bm{I}_n-\bmX(\bmX'\bmX)^{-1}\bmX'-\bme\bme'/\|\bme\|^2\right)\bm{v},
\end{equation}
cannot be the zero vector.
The noise vector considered in this paper is a linear combination of $\bme$ and $\bm{u}$,
given by
\begin{equation}\label{y_ast}
 \bm{\epsilon} =\frac{a\|\bm{e}\|}
{1+b}\left\{\frac{\bm{e}}{\|\bm{e}\|}+ \sqrt{b}\frac{\bm{u}}{\|\bm{u}\|}\right\},
\end{equation}
where $a\neq 0$ and $b\geq 0$.
When $\bmy+\bmep$ is used instead of the original response vector $\bmy$,
we have the following result.

\begin{thm}\label{thm:main}
\begin{enumerate}
\item \label{thm:main:0}
The sample mean of $\bmy+\bmep$ is $\bar{y}$ for any $a$ and $b$.
\item \label{thm:main:1}
The OLS estimator for the response vector $\bmy+\bmep$
remains the same for any $a$ and $b$, that is, 
\begin{equation*}
 (\bmX'\bmX)^{-1}\bm{X}'(\bmy+\bmep)=(\bmX'\bmX)^{-1} \bm{X}'\bm{y}.
\end{equation*}
\item \label{thm:main:2}
The $t$-values for the response vector $\bmy+\bmep$ are given by
\begin{equation*}
 \tilde{t}_j=\left\{\frac{1+b}{1+b+a(a+2)}\right\}^{1/2}t_j,
\end{equation*}
for $j=0,\dots,p$.
\item \label{thm:main:3}
The coefficient of determination for the response vector $\bmy+\bmep$ is
\begin{equation*}
\tilde{R}^2=\left\{\frac{1+b}{1+b+a(a+2)(1-R^2)}\right\}R^2.
\end{equation*}
\item \label{thm:main:4}
The correlation coefficient of $\bmy$ and $\bmy+\bmep$ is
\begin{equation*}
r_{y,y+\epsilon}= \frac{1+b+a(1-R^2)}{(1+b)^{1/2}\{1+b+a(a+2)(1-R^2)\}^{1/2}}.
\end{equation*}
\end{enumerate}
\end{thm}
\begin{proof}
By Part \ref{lem:preparation:ep} of Lemma \ref{lem:preparation}, we have
$ \bmX'\bmep=\bzero$, the first component of which is $\bone'_n\bmep=0$. Hence Part \ref{thm:main:0}
follows.

Since $ \bmX'\bmep=\bzero$, we have 
\begin{equation}\label{same}
(\bmX'\bmX)^{-1}\bm{X}'(\bmy+\bmep)=(\bmX'\bmX)^{-1}\bm{X}'\bmy+(\bmX'\bmX)^{-1}\bm{X}'\bmep
=(\bmX'\bmX)^{-1} \bm{X}'\bm{y}
\end{equation}
which completes the proof of Part \ref{thm:main:1}.

Note that the $t$-values are defined by \eqref{t-value}. By \eqref{same}, 
any component of the OLS estimator keeps the same. 
Further $\sqrt{n-p-1}/d_{j}$ does not depend on the response vector.
Hence 
Part \ref{thm:main:2} follows from Part \ref{lem:preparation:red} of Lemma \ref{lem:preparation}.

Note the coefficient of determination is defined by \eqref{Rsq}. 
Since the sample mean of $ \bmy+\bmep$ is also $\bar{y}$ as in Part \ref{thm:main:0} of this theorem, 
the coefficient of determination for the response vector $ \bmy+\bmep$ is
\begin{equation*}
 1-\frac{\|(\bm{I}_n-\bmX(\bmX'\bmX)^{-1}\bmX')(\bmy+\bmep)\|^2}{
\|\bmy+\bmep-\bar{y}\bm{1}_n\|^2},
\end{equation*}
which is rewritten as
\begin{equation*}
 1-\frac{\left\{1+a(a+2)/(1+b)\right\}\|\bm{e}\|^2}{\|\bm{y}-\bar{y}\bm{1}_n\|^2+\left\{a(a+2)/(1+b)\right\}\|\bm{e}\|^2}
\end{equation*}
by Parts \ref{lem:preparation:dev} and \ref{lem:preparation:red} of Lemma \ref{lem:preparation}. 
By the definition of $R^2$,
we have
\begin{equation}\label{Rsq_1}
 1-R^2=\|\bm{e}\|^2/\|\bmy-\bar{y}\bm{1}_n\|^2,
\end{equation}
which completes the proof of Part \ref{thm:main:3}.

The correlation coefficient of $\bmy$ and $\bmy+\bmep$ is
\begin{equation*}
\frac{ (\bmy-\bar{y}\bm{1}_n)'(\bmy+\bmep-\bar{y}\bm{1}_n)}{\|\bmy-\bar{y}\bm{1}_n\|\|\bmy+\bmep-\bar{y}\bm{1}_n\|}.
\end{equation*}
By Parts 
\ref{lem:preparation:ep} and \ref{lem:preparation:ast}
of Lemma \ref{lem:preparation} as well as \eqref{Rsq_1}, we have
\begin{align*}
 (\bmy-\bar{y}\bm{1}_n)'(\bmy+\bmep-\bar{y}\bm{1}_n) 
&=\|\bmy-\bar{y}\bm{1}_n\|^2+(\bmy-\bar{y}\bm{1}_n)'\bmep \\
&=\|\bmy-\bar{y}\bm{1}_n\|^2+(\hat{\bmy}+\bm{e}-\bar{y}\bm{1}_n)'\bmep \\
&=\|\bmy-\bar{y}\bm{1}_n\|^2+(\bmX\hat{\bmbe}-\bar{y}\bm{1}_n)'\bmep+\bm{e}'\bmep \\
&=\|\bmy-\bar{y}\bm{1}_n\|^2+\bm{e}'\bmep \\
&=\|\bmy-\bar{y}\bm{1}_n\|^2\left[1+\{a/(1+b)\}(1-R^2)\right].
\end{align*}
Further, by Part \ref{lem:preparation:dev} of Lemma \ref{lem:preparation}, we have 
\begin{align*}
 \| \bmy-\bar{y}\bm{1}_n+\bmep\|^2 =\|\bmy-\bar{y}\bm{1}_n\|^2\left[1+\{a(a+2)/(1+b)\}(1-R^2)\right],
\end{align*}
which completes the proof of Part \ref{thm:main:4}.
\end{proof}
The lemma below summarizes fundamental properties related to $\bme$ and $\bmep$,
which are needed in the proof of Theorem \ref{thm:main}.
\begin{lemma}\label{lem:preparation}
 \begin{enumerate}
\item \label{lem:preparation:e}
$\bm{e}$ is orthogonal to $ \bm{1}_n,\bmx_1,\dots,\bmx_p$ or equivalently $\bmX'\bm{e}=\bm{0}$.
\item \label{lem:preparation:u}
$\bm{u}$ is orthogonal to $ \bm{e},\bm{1}_n,\bmx_1,\dots,\bmx_p$
or equivalently $\bmX'\bm{u}=\bm{0}$ and $\bm{e}'\bm{u}=0$.
\item \label{lem:preparation:ep}
$\bmX'\bmep=\bzero$.
\item \label{lem:preparation:ast}
$ \bm{e}'\bmep=a\|\bm{e}\|^2/(1+b)$ and $ \|\bmep\|^2=a^2\|\bm{e}\|^2/(1+b)$.
\item \label{lem:preparation:dev}
The sum of squared deviation of $ \bmy+\bmep$ is 
\begin{align*}
\|\bmy+\bmep-\bar{y}\bm{1}_n\|^2= \|\bm{y}-\bar{y}\bm{1}_n\|^2+\left\{a(a+2)/(1+b)\right\}\|\bm{e}\|^2.
\end{align*}
\item \label{lem:preparation:red}
The residual sum of squares for $ \bmy+\bmep$ is 
\begin{align*}
\|(\bm{I}_n-\bmX(\bmX'\bmX)^{-1}\bmX')(\bmy+\bmep)\|^2=  \left\{1+a(a+2)/(1+b)\right\}\|\bm{e}\|^2.
\end{align*}
\end{enumerate}
\end{lemma}

\begin{proof}
Since $ \bmX'\bmX(\bmX'\bmX)^{-1}\bmX'=\bmX'$, we have
\begin{equation*}
\begin{split}
 \left(\bm{1}'_n\bm{e}  \ \bmx'_1\bm{e} \ \cdots \ \bmx'_p\bm{e} \right)'
&=\bmX'\bm{e} \\
&=\bm{X}'\left(\bm{I}_n-\bmX(\bmX'\bmX)^{-1}\bmX'\right)\bm{y} \\
&=\left(\bm{X}'-\bm{X}'\bmX(\bmX'\bmX)^{-1}\bmX'\right)\bm{y} \\
&=\bm{0},
\end{split}
\end{equation*}
which completes the proof of Part \ref{lem:preparation:e}.
In the same way, Part \ref{lem:preparation:u} can be proved. 

Recall $\bmep$ is given by a linear combination of $\bm{e}$ and $\bm{u}$,
\begin{equation}\label{eq:ep_2}
 \bm{\epsilon} =\frac{a\|\bm{e}\|}
{1+b}\left\{\frac{\bm{e}}{\|\bm{e}\|}+ \sqrt{b}\frac{\bm{u}}{\|\bm{u}\|}\right\}.
\end{equation}
Then Part \ref{lem:preparation:ep} follows 
from Parts \ref{lem:preparation:e} and \ref{lem:preparation:u}.
Part \ref{lem:preparation:ast} follows from the orthogonality of $\bm{e}$ and $\bm{u}$
together with \eqref{eq:ep_2}.  

Since the sample mean of $ \bmy+\bmep$ is $\bar{y}$ by Part \ref{thm:main:0} of Theorem \ref{thm:main},
the sum of squared deviation of $ \bmy+\bmep$, $ \|\bmy+\bmep-\bar{y}\bm{1}_n\|^2$, is
expanded as 
\begin{equation*}
\|\bmy-\bar{y}\bm{1}_n\|^2+2(\bmy-\bar{y}\bm{1}_n)'\bmep+\|\bmep\|^2.
\end{equation*}
By Part \ref{lem:preparation:ep}, we have
\begin{align*}
 (\bmy-\bar{y}\bm{1}_n)'\bmep=(\hat{\bmy}+\bm{e}-\bar{y}\bm{1}_n)'\bmep
=(\bmX\hat{\bmbe}+\bm{e}-\bar{y}\bm{1}_n)'\bmep=\bm{e}'\bmep=a\|\bme\|^2/(1+b).
\end{align*}
Then Part \ref{lem:preparation:dev} follows from Part \ref{lem:preparation:ast}.

Since $ \bmX'\bmep=\bm{0}$ by Part \ref{lem:preparation:ep},
we have
\begin{equation*}
 (\bm{I}_n-\bmX(\bmX'\bmX)^{-1}\bmX')(\bmy+\bmep)=\bme+\bmep.
\end{equation*}
From Part \ref{lem:preparation:ast}, the residual sum of squares is 
\begin{align*}
 \|\bme+\bmep\|^2=\|\bme\|^2+2\bme'\bmep+\|\bmep\|^2
=\|\bme\|^2+2a\|\bme\|^2/(1+b)+a^2\|\bme\|^2/(1+b),
\end{align*}
which completes the proof of Part \ref{lem:preparation:red}. 
\end{proof}
By Theorem \ref{thm:main}, we see that $a=-2$ is a special case, as follows.
\begin{thm}\label{thm:special}
Assume $a=-2$. Then, we have the followings.
\begin{enumerate}
\item \label{thm:special:1}
For any $b>0$, the coefficient of determination for $\bmy+\bmep$ is equal to $R^2$,
the coefficient of determination for the original $\bmy$.
\item \label{thm:special:2}
For any $b>0$, 
the $t$-value of $\beta_j $ $(j=0,1,\dots,p)$ for the response vector $\bmy+\bmep$ is equal to 
$t_j$.
\item \label{thm:special:3}
The correlation coefficient of $\bmy$ and $\bmy+\bmep$ is 
\begin{align}\label{correlation}
 r_{y,y+\epsilon}=1-\frac{2(1-R^2)}{1+b}.
\end{align}
\end{enumerate}
\end{thm}

Recall that 
$\bmep$ is a function of $\bmv$, any random $n$-dimensional vector,
through the relationships, \eqref{def:u} and \eqref{y_ast}, that is,
\begin{align*}
\bm{u}= \left(\bm{I}_n-\bmX(\bmX'\bmX)^{-1}\bmX'-\frac{\bme\bme'}{\|\bme\|^2}\right)\bm{v},\
 \bm{\epsilon} =\frac{a\|\bm{e}\|} 
{1+b}\left\{\frac{\bm{e}}{\|\bm{e}\|}+ \sqrt{b}\frac{\bm{u}}{\|\bm{u}\|}\right\}.
\end{align*}
In Parts \ref{thm:special:1} and \ref{thm:special:2} of Theorem \ref{thm:special},
the choice $a=-2$ guarantees that the coefficient of determination and $t$-value remain the same 
regardless of $\bmv$.

By Part \ref{thm:special:3} of Theorem \ref{thm:special},
$ r_{y,y+\epsilon}$ increases with $b$ for fixed $R^2$.
The correlation coefficients between the original responses $\bmy$
and perturbed responses $\bmy+\bmep$ with $a=-2$, varying $b\geq 0$ and $R^2$,
are illustrated in Table \ref{tab:0}. 

In actual application, it is desirable to have
relatively high correlation, because data users might assume that the perturbed
response is close to the original response. However, if the correlation is very high, then
the perturbed response is very close to the original response, and the objective of
concealing the actual response cannot be achieved. Thus, it is necessary to determine a
value of $b$ that prevents the perturbed response from being too close to the actual
response, as will be discussed through the analysis of real data in the next section.

\begin{remark}\label{rem:a2}
 When $a=-2$ and $b=0$, we have $\bmep=-2\bme$ as the noise or, equivalently
\begin{align}
\bmy-2\bme=\hat{\bmy}-\bme 
\end{align}
as the perturbed response. 
In this case, it is clear that
the coefficient of determination and $t$-value remain the same, since
$y_i$ and $y_i-2e_i$ for $i=1,\dots,n$ are symmetric with respect to the point $\hat{y}_i=y_i-e_i$.
Since the noise $\bmep=-2\bme$ does not depend on $\bmv$, there is no randomness in the noise.
Theorem \ref{thm:special} ensures that, for random $\bmv$, as in \eqref{y_ast},
 it is possible to construct the noise $\bmep$ such that the coefficient of
 determination and $t$-value remain the same.
\end{remark}

\begin{table}
\caption{Correlation coefficient of $\bmy$ and $\bmy+\bmep$ with $a=-2$} 
\label{tab:0}
\begin{tabular}{cccccccccc}\toprule
$R^2\backslash b$ & $0$ & $0.25$ & $0.5$ & $0.75$ & $1.0$ & $1.25 $ & $1.5$ & $1.75$ & $2.0$ \\ \midrule
$0.4$ & -0.2 & 0.04 & 0.2 & 0.31& 0.4 & 0.47 & 0.52 & 0.56 & 0.6 \\
$0.6 $ & 0.2 & 0.36 & 0.47 & 0.54 & 0.6 & 0.64 & 0.68 & 0.71 & 0.73 \\
$0.8$ & 0.6 & 0.68 & 0.73 & 0.77 & 0.8 & 0.82 & 0.84 & 0.85 & 0.87 \\ \bottomrule
\end{tabular}
\end{table}

\begin{remark}\label{rem:not2}
As in Theorem \ref{thm:special}, 
the choice $a=-2$ with random $\bmv$ was surprisingly found to retain the $R^2$ and $t$ values. 
Following are some remarks for the other choices.
For $a\in(-\infty,-2)\cup(0,\infty)$, both $R^2$ and the absolute value of $t$ values are
reduced.
For example, $b>0$ and $a=-1\pm \sqrt{b+2}\in(-\infty,-2)\cup(0,\infty)$ yield
\begin{align}\label{eq:reduction}
 \tilde{t}_j=\frac{1}{\sqrt{2}}t_j, \tilde{R}^2=\frac{R^2}{2-R^2}<R^2.
\end{align}
Note that $R^2$ and $t$ values can be completely controlled.
Thus the data provider safely provide data with the relation between $ \{t_j,R^2\}$
and $\{\tilde{t}_j, \tilde{R}^2\}$ described by \eqref{eq:reduction},
and practitioners can restore the original $R^2$ and $t$-value independently.
An efficient method of opening data with reduced accuracy will be reported elsewhere.
\end{remark}

\section{Numerical experiment}
\label{sec:numerical}
In the previous section, a method was proposed to add noise to the response variable.
This method can be applied when a real estate database is released into the public
domain, by adding noise to the transacted price, which is considered to be sensitive
information in Japan. As Theorem \ref{thm:special} in the previous section ensures, the results of
regression analysis using the perturbed data will not change. However, in actual
application, a variety of analyses will be devised, and the theorems do not apply in
cases with unexpected applications. Thus, it is necessary to verify whether the
proposed method remains appropriate even in unexpected applications.

The precision of the results might be degraded if analytical operations not assumed in
the theory are applied. In such cases, permissible error levels resulting from the
perturbation must be determined. In the following, a numerical experiment to
determine the relationship between perturbation and precision level is discussed.

\subsection{Data used in the experiment}
The data source used for the numerical experiment was At Home Co.~Ltd. The data
contained real estate advertisement information from $2008$. The database for the
experiment was created by supplementing some spatial variables. It contained $1,320$
cases of newly built detached houses in Setagaya Ward in Tokyo Prefecture\footnote{We selected data that contained information about the designated floor area
and building coverage ratios. Such data are thought to be important in real estate
analysis in Japan. In the original database, a new record was added each time a
property owner changed the price in the advertisement. In such situations, we selected
only the newest record.}.
The variables included the price of the property (yen), the time to the nearest railway
station (minutes), a dummy variable representing bus usage, the area of the site
(square meters), the floor area (square meters), a dummy variable signifying leased
land, the designated building coverage ratio, the designated floor area ratio, the time to
Shinjuku by rail from the nearest station (minutes), the time to Shibuya by rail
(minutes), the time to Yokohama by rail (minutes), the time to Tokyo by rail (minutes),
the width of the nearest road (meters), and a dummy variable signifying the nearest
road to the south of lot. Note that Shinjuku, Shibuya, Yokohama, and Tokyo are four major
railway stations in the study region. Among these variables, the times to the railway
stations, width of the nearest road, and dummy variable signifying the nearest road to
the south are spatial variables, as described in the next subsection.

\subsection{Creation of spatial variables}
The times to the major railway stations from the nearest station; the width of the
nearest road, as measured from the representative point of the property; and the
dummy variable signifying whether the nearest road is located to the south of the
property were added to the original database as variables for signifying spatial
relationships. The width of the nearest road to the representative point of the property
was regarded as the width of the nearest road, which was done because precise digital
data for lots are not available. Accordingly, the dummy variable signifying whether the
nearest road was located to the south of the property was regarded as the dummy
variable signifying the nearest road to the south of lot.

The times to the major railway
stations from the nearest station were calculated using the search system for guiding
transferring railways provided by NAVITIME Japan Co.~Ltd. This system
automatically calculates the time required to travel to the major railway stations, i.e.,
Shinjuku, Shibuya, Yokohama, and Tokyo, from the railway station nearest to the
property. To determine the times required in this study, the departing time was set to
$12:00$ (noon) on August $2$, $2010$.

The width of the nearest road from the representative point of the property was
calculated as follows. Mapple $10000$ digital data produced by Shobunsha Publications
Inc. contain digital road data classified by road width categories, such as $4$-$5$m and
$5$-$6$m. The median of each class was assigned as the road width. For example, a width of
$4.5$m was used for the $4$-$5$m class. With the geographic information system (GIS)
software ArcGIS $10$, the nearest road was assigned for each property, and the width of
the road calculated as described above was set to be the width of road nearest to the
representative point of the property.

In the real estate market in Japan, a residential lot tends to be evaluated highly if it is
adjacent to a road to the south of lot, because receiving substantial sunlight is preferred in
Tokyo. For example, \cite{RealEstate1986} treats
properties adjacent to roads to the south of lots more favorably in their property appraisals.
With this preference in mind, the dummy variable signifying whether the nearest road
is located to the south of the property was also added to the database.

This dummy variable was constructed as follows. Using ArcGIS $10$, the direction to the
nearest road was calculated, such that $0^\circ$ was located to the east, and the value
increased to $180^\circ$ counterclockwise and decreased to $-180^\circ$ clockwise.
The range from $-135^\circ$ to $-45^\circ$ was judged to be to the south, in which case the dummy variable was set to one, and it was set to zero otherwise.
The statistics of the variables are summarized in Table \ref{tab:1}.

\begin{table}
\caption{Summary of variable statistics} 
 \label{tab:1}
\begin{threeparttable}
\begin{tabular}{ccccc}\toprule
& min & max & mean & s.d. \\ \midrule
price of the property (yen) & 34800000 & 330000000 & 72431491 & 25539447 \\
time to the nearest railway station (minutes) & 0 & 25 & 10.60 & 4.83 \\
d.v.\tnote{a} \ representing bus usage & 0 & 1 & 0.07 & 0.26 \\
area of the site (square meters) & 29.53 & 211.49 & 88.56 & 25.48 \\
floor area (square meters) & 47.07 & 228.48 & 98.94 & 20.06 \\
d.v.\tnote{a} \ signifying leased land & 0 & 1 & 0.03 & 0.17 \\
designated building coverage ratio & 40 & 80 & 54.18 & 7.70 \\
designated floor area ratio& 80 & 300 & 141.43 & 47.10 \\
time to Shinjuku by rail (minutes)& 5 & 32 & 18.72 & 5.29 \\
time to Shibuya by rail (minutes)& 3 & 29 & 14.86 & 6.01 \\
time to Yokohama by rail (minutes) & 17 & 64 & 44.30 & 10.99 \\
time to Tokyo by rail (minutes)& 23 & 48 & 34.09 & 4.90 \\
width of the nearest road (meters) & 4.5 & 35 & 5.80 & 2.25 \\
d.v.\tnote{a} \ signifying the nearest road to the south of lot& 0 & 1 & 0.28 & 0.45 \\ \bottomrule
\end{tabular}
\begin{tablenotes}\footnotesize
\item[a] d.v. stands for ``dummy variable''.
\end{tablenotes}
\end{threeparttable}
\end{table}

\subsection{Numerical experiment with perturbed property price}
The perturbed property price, which was generated by adding noise to the response
variable using the method described in the previous section, was numerically tested as described in this subsection. The explanatory variables used were the $13$ variables in Table \ref{tab:1}.

\subsubsection{Statistics of the perturbed property price}
Although Part \ref{thm:main:0} of Theorem \ref{thm:main} guarantees that the mean of the perturbed response
variable is exactly equal to the mean of the original response, the equality or similarity
of the other statistics, such as the minimum value, maximum value, and first and third
quantiles, theoretically cannot be controlled. In this section, the generation of four sets
of quasi-response variables with different $\bmv$ values, $a=-2$, and $b=1$ is described, to
analyze the degrees of perturbation of the statistics among the five sets, including the
original response (original, quasi1, quasi2, quasi3, and quasi4).

Figure \ref{fig:boxplot} shows boxplots of the five sets. When the original and quasi-response
variables are compared, the medians are very similar, but the quantiles, minima, and
maxima are quite different. It is also evident that, among the four sets of quasi
variables, all of the statistics are similar.
Figure \ref{fig:y_orig_quasi_cor} shows scatterplots of the original
variables and of the four sets of quasi variables. Although the plots for the four sets of
quasi variables appear very similar, the different $\bmv$ values imply different quasi
variables, as explained in Section \ref{sec:theory}.
Table \ref{tab:3} provides the correlation matrix for the five sets of response variables.
By Part \ref{thm:special:3} of Theorem \ref{thm:special}, the correlation coefficient between
the original and quasi variables is given theoretically by $1-2(1-R^2)/(1+b)$, which
equals $R^2$ for $b=1$. Among the quasi variables, the correlations in all cases are approximately $0.78$.

\begin{figure}
\centering
\includegraphics[width=14cm, height=11cm]{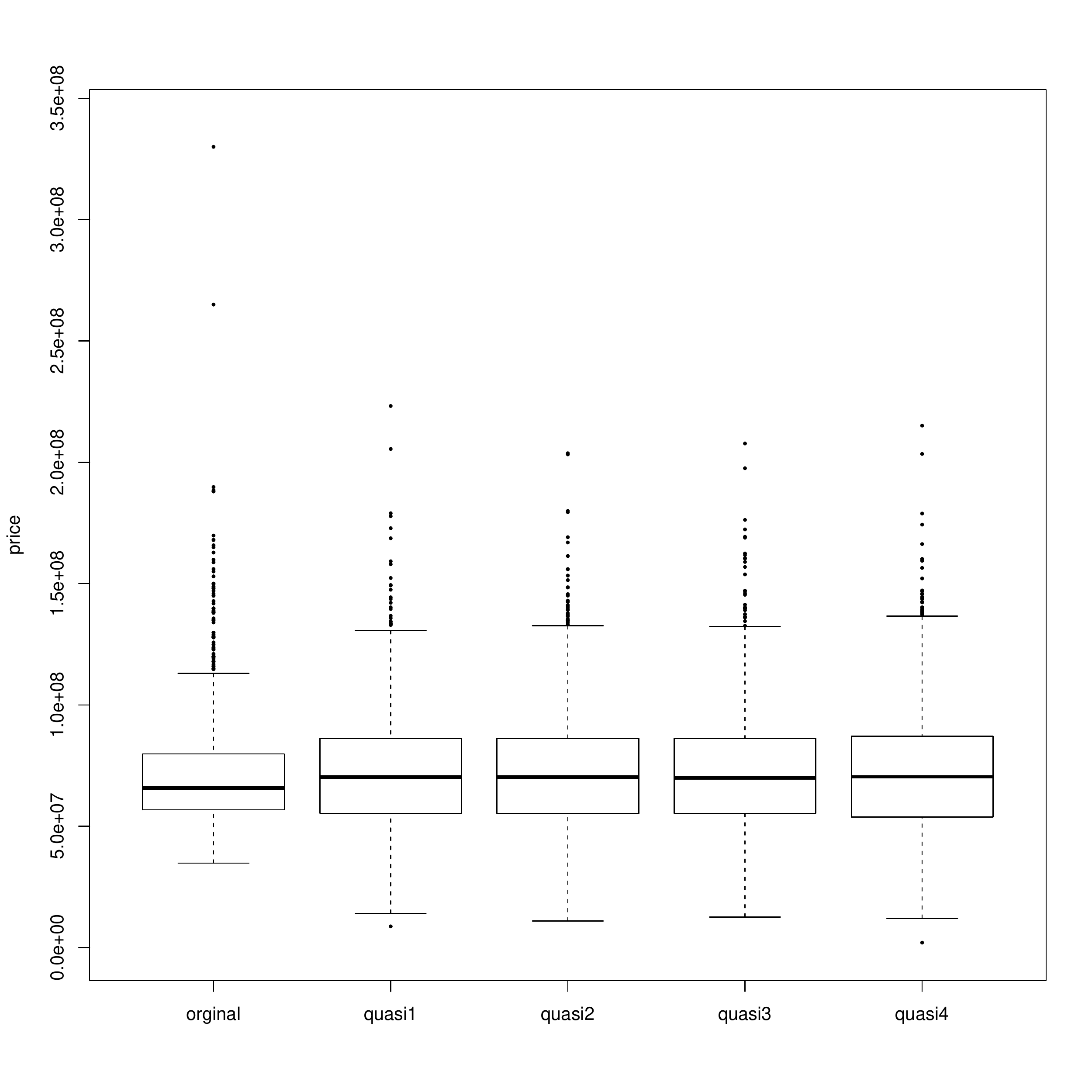}
\caption{Boxplots of five sets of response variables}
\label{fig:boxplot}
\end{figure}

\begin{figure}
\centering
\includegraphics[width=14cm, height=16cm]{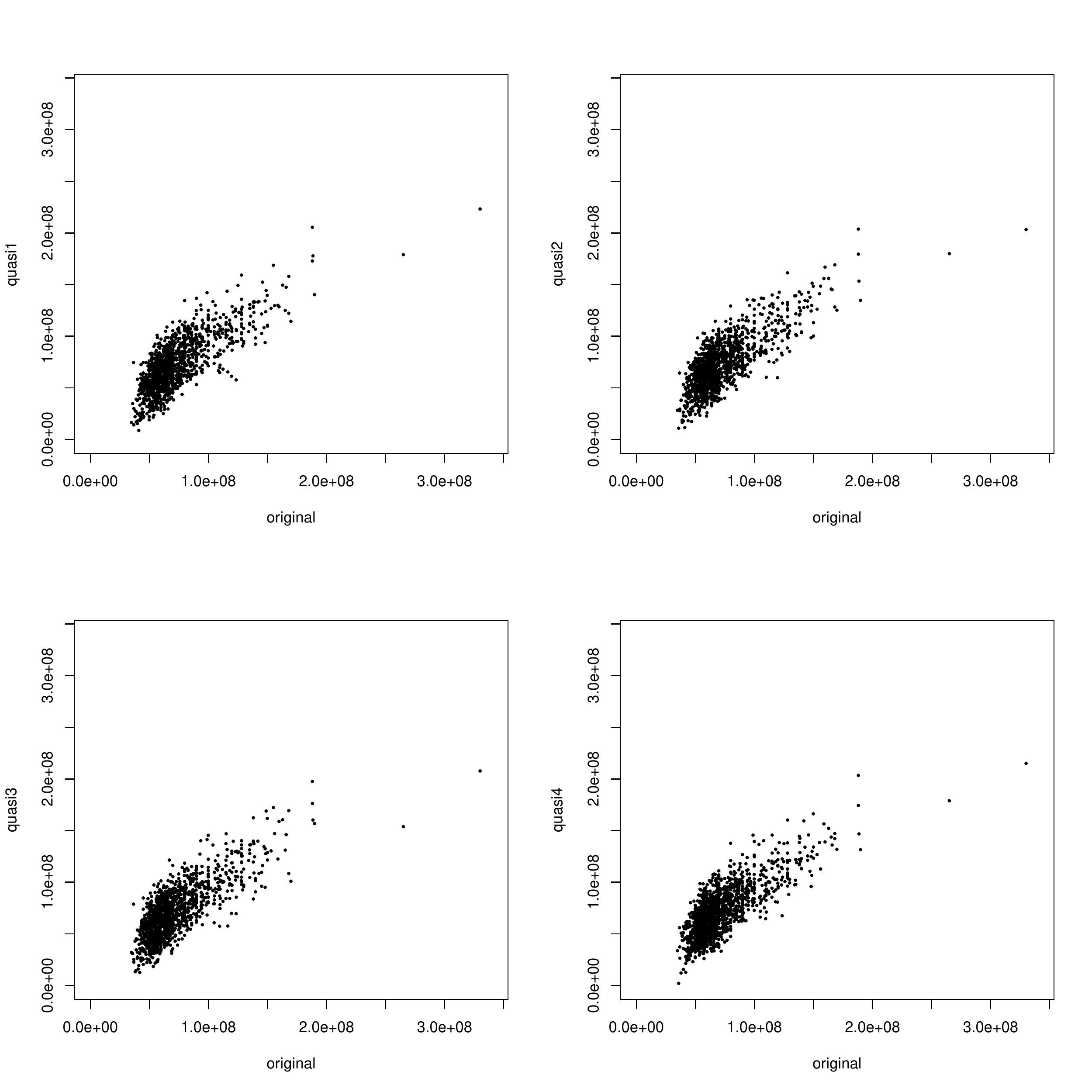}
\caption{Scatterplots of five sets of response variables}
\label{fig:y_orig_quasi_cor}
\end{figure}

\begin{table}
\caption{Correlations between original response and four sets of quasi responses} 
\label{tab:3}
\begin{tabular}{cccccc}\toprule
&	orig & quasi1 & quasi2 & quasi3 & quasi4 \\ \midrule
orig & 1 & 0.7748 & 0.7748 & 0.7748 & 0.7748 \\
quasi1 & 0.7748 & 1 & 0.7693 & 0.7749 & 0.7814 \\
quasi2 & 0.7748 & 0.7693 & 1 & 0.7870 & 0.7812 \\
quasi3 & 0.7748 & 0.7749 & 0.7870 & 1 & 0.7733 \\
quasi4 & 0.7748 & 0.7814 & 0.7812 & 0.7733 &1 \\ \bottomrule
\end{tabular}
\end{table} 

\begin{remark}\label{rem:minus}
In this particular data set, the response variable was the property price,
which was expected to be positive. Hence, a positive perturbed price is strongly
desirable. As claimed in Remark \ref{rem:a2}, for sufficiently small $b$, we have
\begin{align*}
 \bmy+\bmep\approx \hat{\bmy} - \bme.
\end{align*}
Suppose there exist individuals $i$ with relatively expensive prices $y_i$, when relatively lower
prices $y_i$ are expected. Then $e_i$ increases, and as a result
\begin{align}\label{eq:minus}
 y_i+\epsilon_i \approx \hat{y}_i  - e_i<0
\end{align}
can occur. In our data set, such situations rarely occurred for $b=1.2$ or less and never
occurred for $b=1.3$ or greater. To the best of our knowledge, the occurrence of such
situations is theoretically not controllable through the choices of $b$ and $\bmv$.
When \eqref{eq:minus} occurs, it is recommended to generate $\bmep$ with different $\bmv$
values until $\min \{y_i+\epsilon_i\}>0$ is achieved.
\end{remark}

\subsubsection{Regression analysis using only a portion of the database}
The theory assumes that all of the data will be used for the analysis. If only a portion of
the perturbed data is used, then the theorems do not apply exactly. In actual analyses
for real estate data, only a portion of the (perturbed) database is used for the analysis.
In such cases, it is necessary to know how the results might differ from the theoretical
results and to follow the subsequently described guidelines to choose an appropriate
value of $b$.

For this purpose, a critical value of $b$ may be obtained such that the difference between
the regression model using the perturbed property price as the response variable and
the original property price is not statistically significant.

\subsubsection{Chow test}
From $1,320$ cases (total database), $20\%$ (i.e., $264$ cases) were selected randomly, and
perturbed prices were generated for $13$ $b$ values (i.e., $b = 0.5$, $0.6$, $0.7$, $0.8$, $0.9$, $1.0$, $1.1$, $1.2$, $1.3$, $1.4$, $1.5$, $2.0$, $2.5$).
The Chow test was applied to determine whether the regression models with the original and the perturbed prices could be regarded as the same model.
For each value of $b$, $1,000$ independently chosen samples were created and analyzed.
As a result, for each value of $b$, $1,000$ values of the Chow test $F$ value were derived.
Ordering these values by magnitude, $5\%$, $10\%$, $50\%$, $90\%$ and $95\%$
(i.e., the $50^{\text{th}}$, $100^{\text{th}}$, $500^{\text{th}}$, $900^{\text{th}}$ and $950^{\text{th}}$ value) of the points of the $F$ value were derived.
Figure \ref{fig:F} shows that larger values of $b$ correspond to smaller $F$-value variations.
Given the objective of choosing an appropriate value of $b$ to generate a properly perturbed
property price, the minimum value of $b$ for which the null hypothesis of the Chow test
(namely, $H_0$:``There is no statistically significant difference between two models'') is not
rejected can be considered the critical value of $b$. Note that the $F$ value in the $F$
distributions with degrees of freedom $14$ and $500$ that achieves a significance level $0.05$
is $F=1.71$. Hence, if $F$ is less than $1.71$, then the null hypothesis cannot be rejected,
and therefore the two models can be regarded as statistically the same.

For each value of $b$, the percentage of $F$ values among $1,000$ trials that satisfied the
acceptance condition of $F$ less than $1.71$ was calculated.
In our numerical experiment,
these percentages are $65.0\%$, $97.0\%$, and $100\%$ for $b=0.5$, $b=1.0$, and $b\geq 1.4$,
respectively, as seen in Table \ref{tab:2}.

\begin{figure}
\centering
\includegraphics[width=12cm, height=8cm]{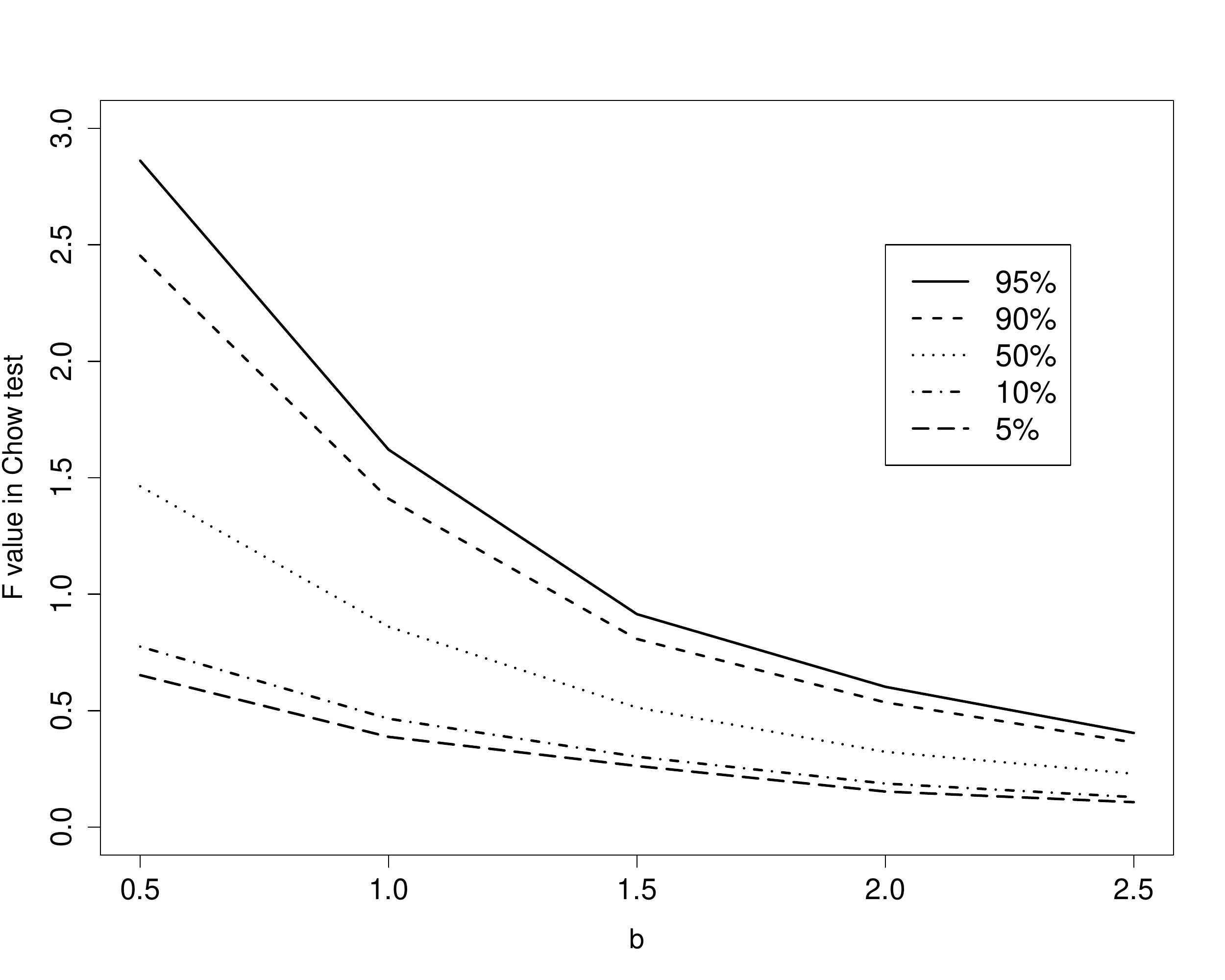}
\caption{The relation between F value and $b$}
\label{fig:F}
\end{figure}

\subsubsection{Recommended standard for $b$ value}
In the numerical experiment described above, when $20\%$ was selected randomly, the
Chow tests used to test the identity of the two models, that is, the regression models
with the actual and perturbed property prices as the response variables, demonstrated that the F value satisfied the acceptance condition with $97.0\%$ probability when $b = 1.0$
and $100\%$ probability when $b\geq 1.4$.
Assuming that approximately $5\%$ is the permissible level for hypothesis rejection
(i.e., that two models cannot be regarded as the same),
$b=1.0$ is judged appropriate, as it ensures that the perturbed price is perturbed
sufficiently and, nonetheless, that the regression model with the perturbed price can be
regarded as identical to the regression model with the original price.
The appropriate value of $b$ differs if another percentage is used to select the sample.
For instance, we let $q$ be the percentage used to select the sample and, using the above
numerical experiment, we let $q = 0.2$ ($20\%$).
Assuming a $5\%$ rejection level, the critical value of $b$, $b_*$,
such that for $b$ less than $b_*$ the probability of rejection becomes greater
than $5\%$, was calculated by changing $q$.
Table \ref{tab:2} summarizes the results. For all $q$ values investigated in this study,
$b_* = 1.0$ appears to be a reasonable choice, as it balances
the similarity and the perturbation to the original price.


\begin{table}
\caption{Percentages of samples that accepted the null hypothesis, for which the perturbed sample can be regarded as statistically identical to the original sample for sample selection percentage, $q$, and $b$ value} 
\label{tab:2}
\begin{tabular}{ccccccccccc}\toprule
$b\backslash q$ & 0.05 & 0.10 & 0.20 & 0.30 & 0.40 & 0.50 & 0.60 & 0.70 & 0.80 & 0.90 \\ \midrule
0.5 & 0.591 & 0.576 & 0.650 & 0.728 & 0.827 & 0.918 & 0.969 & 0.992 & 0.996 & 1.000 \\
0.6 & 0.686 & 0.670 & 0.744 & 0.808 & 0.903 & 0.947 & 0.988 & 0.997 & 0.999 & 1.000 \\
0.7 & 0.729 & 0.764 & 0.808 & 0.873 & 0.943 & 0.976 & 0.994 & 1.000 & 0.999 & 1.000 \\
0.8 & 0.815 & 0.845 & 0.858 & 0.928 & 0.966 & 0.988 & 0.996 & 1.000 & 1.000 & 1.000 \\
0.9 & 0.867 & 0.867 & 0.931 & 0.966 & 0.989 & 1.000 & 0.998 & 1.000 & 1.000 & 1.000 \\
1.0 & 0.930 & 0.925 & 0.970 & 0.987 & 0.995 & 0.998 & 1.000 & 1.000 & 1.000 & 1.000 \\
1.1 & 0.960 & 0.968 & 0.983 & 0.998 & 0.997 & 1.000 & 1.000 & 1.000 & 1.000 & 1.000 \\
1.2 & 0.978 & 0.987 & 0.994 & 0.999 & 0.999 & 1.000 & 1.000 & 1.000 &1.000 & 1.000 \\
1.3 & 0.994 & 0.994 & 0.995 & 1.000 & 1.000 & 1.000 & 1.000 & 1.000 & 1.000 & 1.000 \\
1.4 & 0.996 & 1.000 & 1.000 & 1.000 & 1.000 & 1.000 & 1.000 & 1.000 & 1.000 & 1.000 \\
1.5 & 0.999 & 0.998 & 1.000 & 1.000 & 1.000 & 1.000 & 1.000 & 1.000 & 1.000 & 1.000 \\ \bottomrule
\end{tabular}
\end{table}

\section{Conclusion}
\label{sec:conclusion}
This paper proposed a new method of perturbing a major variable by adding noise,
while ensuring that the results of regression analysis are not affected. The extent of the
perturbation can be controlled using a single parameter, $b$, which eases actual
perturbation application. Moreover, $b=1.0$ can be regarded as an appropriate value for
achieving both sufficient perturbation to mask the original values and sufficient
coherence between the perturbed and original data.

The proposed method masks only one major variable, but in actual application, many
situations may be encountered in which only one variable is critical to put the entire
dataset in the public domain. Our method will be useful in such situations. There are
other possible uses of perturbed data, and the appropriateness of the $b$ value must be
examined by testing a greater variety of data-use cases. Admittedly, application of the
proposed method is limited, because other variables are assumed to retain their
original values. Thus, further methods of perturbing the explanatory variables are
necessary to broaden the range of applications. Such extensions will be provided in
subsequent work.

\section*{Acknowledgements}
This work was partly supported by a Research Grant from the Secom Science and
Technology Foundation. The work of the first and third authors was partly supported
by Grant-in-Aids for Scientific Research No.~25330035 and No.~26590036, respectively,
from the Japan Society for the Promotion of Sciences.

\end{document}